\documentclass[a4paper,11pt]{article}
\usepackage[utf8]{inputenc}
\usepackage{amsthm, amssymb, mathrsfs, amscd, amsmath, stmaryrd}
\usepackage{fancyhdr}
\usepackage{mathtools}
\usepackage{tikz}
\usepackage[pdfa]{hyperref}
\hypersetup{               
	colorlinks=true,                
	breaklinks=true,                
	urlcolor= black,                 
	bookmarksopen=false,
	filecolor=black,
	citecolor=black,
	linkbordercolor=blue
}
\usepackage{graphicx}
\usepackage{indentfirst} 
\usepackage{multirow}
\usepackage{latexsym}
\usepackage{longtable}
\usepackage{verbatim}
\usepackage{amsfonts}
\usepackage{graphics}
\usepackage{mathrsfs}
\usepackage{enumerate}
\usepackage{cite} 
\usepackage{url}
\usepackage{color,soul}
\usepackage{titling}
\usepackage{complexity}
\usepackage[mm, ff, primitives]{cryptocode}
\usepackage{float}

\providecommand{\keywords}[1]
{
	\small	
	\textbf{\textit{Keywords---}} #1
}

\DeclarePairedDelimiter\flor{\lfloor}{\rfloor}
\newcommand{\floor}[1]{\flor*{#1}}

\newcommand{\qquote}[1]{``#1''}

\newcommand{\verifyx}{\pcalgostyle{Verify}}
\newcommand{\ME}{\pcalgostyle{ME}}
\newcommand{\MME}{\pcalgostyle{MME}}
\newcommand{\LY}{\pcalgostyle{LYHW19}}

\newcommand{\pk}{\pcalgostyle{pk}}
\newcommand{\sk}{\pcalgostyle{sk}}

\newcommand{\Gpub}{G_{\text{pub}}}

\theoremstyle{plain}
\newtheorem{Th}{Theorem}
\newtheorem{Prop}[Th]{Proposition}

\theoremstyle{definition}

\title{A note on a Code-Based Signature Scheme}
\author{Giuseppe D'Alconzo\\ \href{mailto:giuseppe.dalconzo@polito.it}{giuseppe.dalconzo@polito.it} \\ Department of Mathematical Sciences, Politecnico di Torino}
\date{}

\fancyhfoffset[L]{1cm} 
\fancyhfoffset[R]{1cm} 
\begin{document}
	\maketitle
	
	\begin{abstract}
		In this work, we exploit a serious security flaw in a code-based signature scheme from a 2019 work by Liu, Yang, Han and Wang. They adapt the McEliece cryptosystem to obtain a new scheme and, on top of this, they design an efficient digital signature. We show that the new encryption scheme based on McEliece, even if it has longer public keys, is not more secure than the standard one. Moreover, the choice of parameters for the signature leads to a significant performance improvement, but it introduces a vulnerability in the protocol.
	\end{abstract}
	
	
	\section{Introduction}
	\label{Sect:intro}
	
	\paragraph{Post-Quantum Cryptography.}
	With the emerging menace of quantum computation, there is the urge of replacing cryptosystems used nowadays, based on the Factorisation Problem and the Discrete Logarithm Problem, with quantum-resistant alternatives. Because of this, in 2016 the National Institute of Technology and Security (NIST) started a call to evaluate and standardise quantum-resistant cryptosystems\footnote{NIST Post-Quantum Standardization process webpage: \url{https://csrc.nist.gov/Projects/post-quantum-cryptography/post-quantum-cryptography-standardization}, Accessed: 2022-06-16}. There are two categories of primitives under evaluation: Key-Encapsulation Mechanisms (KEM) and Digital Signatures. The standardisation process is still ongoing and counts many schemes based on different assumptions: lattices, linear codes and multivariate polynomials, among others.
	
	\paragraph{Code-Based Signatures.}
	For Key-Encapsulation Mechanism one of the finalists of the NIST's call is based on linear codes (Classic McEliece \cite{bernstein2017classic}), while for signatures the situation is very different: there are no code-based schemes among the finalists or alternative candidates. Previous trapdoor schemes are either not secure, for example KKS \cite{kabatianskii1997digital}, or unpractical, like CFS \cite{courtois2001achieve}. In particular, a big effort has been done to improve the performance of CFS-like schemes, leading to dangerous vulnerabilities, like the one exploited in this work or in \cite{d2021security}. However, new constructions are emerging \cite{debris2019wave,barenghi2021less} and an innovative current of research regarding Zero-Knowledge proofs and the \qquote{MPC-in-the-head} paradigm is very promising \cite{feneuil2021shared,gueron2022designing} and it seems be the new frontier of code-based signatures.
	
	\paragraph{Our contribution.}
	In this work, we cryptanalyze the two schemes presented in \cite{liu2019secure}: a public-key cryptosystem, a modification of McEliece, that we call Modified McEliece ($\MME$), and a digital CFS-like signature using the $\MME$ scheme denoted with $\LY$. We show that the new $\MME$ is not more secure than the original McEliece. In \cite{liu2019secure}, they claim that the signature $\LY$ is more efficient than the CFS scheme. We show that for proposed parameters, this scheme can be broken in practical time, while using bigger parameters to establish security, we have a fall in the performance.\\
	This work is organized as follows: after recalling some preliminaries in Section \ref{Sect:Prel}, in Section \ref{Sect:MME} we introduce the new public-key encryption scheme $\MME$ and we analyze its security with respect to the textbook McEliece cryptosystem. In Section \ref{Sect:SS} we recall the digital signature $\LY$ and we show how to recover, using the proposed parameters, the secret key from the public key.

\section{Preliminaries}
\label{Sect:Prel}

\paragraph{Notation.}Let $\mathbb{N}=\{1,2,\dots\}$ and $\mathbb{R}$ be the sets of natural and real numbers, respectively. We denote with $\lambda$ the security parameter. A function $\epsilon:\mathbb{N}\to\mathbb{R}$ is \emph{negligible} if there exists $n_0$ such that for every $n>n_0$ we have $\epsilon(n)\le {1}/{p(n)}$ for every polynomial $p$. A function not having this propriety is called \emph{non-negligible}. A probability is \emph{overwhelming} if it is equal to $1-\epsilon$, where $\epsilon$ is a negligible function. For a prime power $q$, $\mathbb{F}_q$ is the finite field with $q$ elements, and $\left(\mathbb{F}_q\right)^n$ is the $n$-dimensional vector space over $\mathbb{F}_q$. The \emph{Hamming weight} of a vector $x$ is the number of its non-zero coordinates, and its denoted with $\mathrm{w}(x)$. With $||$ we denote the concatenation of strings or vectors.

\paragraph{Linear Codes and Goppa Codes.}A $\left[n,k,d\right]$ \emph{linear code} $\mathcal{C}$ over $\mathbb{F}_q$ is a $k$-dimensional vector subspace of $\left(\mathbb{F}_q\right)^n$ such that $d=\min_{x\in\mathcal{C}\setminus\{0\}}\{w(x)\}$. Parameters $n,k$ and $d$ are respectively called the \emph{length}, the \emph{dimension} and the \emph{minimum distance} of the code $\mathcal{C}$. Given a basis $\mathcal{B}$ of $\mathcal{C}$, a \emph{generator matrix} for $\mathcal{C}$ is a $k\times n$ matrix with coefficients in $\mathbb{F}_q$ having elements of $\mathcal{B}$ as rows.
The \emph{error correction capability} of a code is given by $t=\floor{\frac{d-1}{2}}$. Given a vector $c$ in $\mathcal{C}$ and a vector $e$ in $\left(\mathbb{F}_q\right)^n$ of weight at most $t$, the \emph{decoding} of a given $y=c+e$ is the procedure of recovering $c$ and $e$. For the rest of this work, we use the convention that given a decoding algorithm $\textbf{D}_{\mathcal{C}}$ and a vector $y=c+e$, the decoding of $y$ is given by $c$ in $\mathcal{C}$. For a random code, the decoding is an hard problem \cite{berlekamp1978inherent}, but there exist families of codes having efficient decoding algorithms. 

For cryptographic constructions, a relevant family of efficiently decodable linear codes is given by \emph{binary irreducible Goppa codes} \cite{berlekamp1973goppa}. A binary irreducible Goppa code is a code over $\mathbb{F}_{2^m}$ defined by a monic irreducible polynomial $g(X)$ in $\mathbb{F}_{2^m}[X]$ of degree $t$ and by an ordered set $L=\{\alpha_1,\dots,\alpha_n\}$ in $\mathbb{F}_{2^m}$ called \emph{support}. Usually, we set $n=2^m$, i.e. the support is the whole field $\mathbb{F}_{2^m}$. The Goppa code $\Gamma(g,L)$ is a $[n,n-mt,d\ge 2t+1]$ linear code over $\mathbb{F}_{2^m}$. Given the polynomial $g(X)$ and the support $L$, there is a polynomial time decoding algorithm $D_{\Gamma(g,L)}$. For this work we do not need more details on the construction of $\Gamma(g,L)$ or on the decoding algorithm $D_{\Gamma(g,L)}$; the interested reader can see \cite{berlekamp1973goppa,patterson1975algebraic,macwilliams1977theory}.

\paragraph{Public-Key Encryption Schemes.}
A \emph{public-key encryption scheme} is a tuple of polynomial-time algorithms $\left(\kgen,\enc,\dec\right)$: $\kgen$ takes in input a security parameter $\lambda$ in unary and returns a pair of public-secret keys, while $\enc$ and $\dec$ are the encryption and decryption algorithms. We want that for every $\lambda$, given $\left(\pk,\sk\right)=\kgen(1^\lambda)$, the following holds for every suitable plaintext $m$
$$ \dec(\pk,\sk,\enc(\pk,m)) = m $$
with overwhelming probability.

A \emph{key-recovery forger} $\mathcal{F}$ for a public-key encryption scheme is a probabilistic polynomial-time algorithm that, on input the public key $\pk$, returns the secret key $\sk$ with non-negligible probability.

Many code-based public key encryption schemes can be found in literature: the most notable are the McEliece (Section \ref{Sect:MME}) \cite{mceliece1978public} and the Niederreiter \cite{niederreiter1986knapsack} cryptosystems.

\paragraph{Digital Signatures and CFS.}
A digital signature algorithm is a scheme composed by three algorithms $\left(\kgen,\sig,\verifyx\right)$: $\kgen$ takes in input a security parameter $\lambda$ in unary and returns a pair of public-private keys, while $\sig$ and $\verifyx$ are the signature and verify algorithms. If $\left(\pk,\sk\right)=\kgen(1^\lambda)$ and $\sigma=\sig(\sk,m)$ , then $\verifyx(\pk,\sigma,m)$ accepts the signature $\sigma$ with overwhelming probability for every message $m$ and security parameter $\lambda$.

We adapt the key-recovery forger in the case of digital signatures. A \emph{key-recovery forger} $\mathcal{F}$ for a signature scheme is a probabilistic polynomial-time algorithm that, on input the public key $\pk$, returns the secret key $\sk$ with non-negligible probability.

The CFS scheme \cite{courtois2001achieve} is a code-based digital signature based on the \emph{hash-and-sign} paradigm. The public key is a randomly scrambled binary irreducible Goppa code $\mathcal{C}$, while the secret key is the decoding algorithm $D_{\mathcal{C}}$. This scheme is highly unpractical: in the signing algorithm, we compute the hash $h(m,i)$, where $m$ is the message and $i$ a nonce, until the digest is a decodable vector for $\mathcal{C}$. This operation requires a huge number of hashes and decodings. The construction presented in \cite{liu2019secure} is based on CFS and a modified version of McEliece, to decrease the computational effort of the signing algorithm.

\section{Modified McEliece Cryptosystem}
\label{Sect:MME}

In this section we recall the textbook definition of the McEliece public-key cryptosystem, showing the modified version from \cite{liu2019secure} and proving that, even if public keys are longer, this new version does not improve the security of the scheme.

\subsection{McEliece and the Modified Version}
\label{SubS:ME}
The following is the textbook version of McEliece \cite{mceliece1978public}, the starting point of the modification given in \cite{liu2019secure}. Alternative versions of this scheme use different techniques to achieve higher levels of security and/or more compact keys. The public-key encryption cryptosystem McEliece ($\ME$) is composed of the following three algorithms.
\begin{itemize}
	\item $\kgen_{\ME}(1^\lambda)$: generate a Goppa code $\mathcal{C}$ over $\mathbb{F}_{2^m}$ with parameters $[n=2^m,k,2t+1]$ according to $\lambda$. Let $G$ be a generator matrix of $\mathcal{C}$ and $\textbf{D}_{\mathcal{C}}$ an efficient decoding algorithm. Sample two random matrices with coefficients in $\mathbb{F}_{2^m}$: a $k\times k$ invertible matrix $S$ and a $n\times n$ permutation matrix $P$. Set $\Gpub=SGP$. Return $(\Gpub,t)$ as public key and $(S,P,\textbf{D}_{\mathcal{C}})$ as secret key.
	
	\item $\enc_{\ME}(m,(\Gpub,t))$: the ciphertext of a message $m$ in $\mathbb{F}_2^{k}$ is given by $m\Gpub+e$, where $e$ is a randomly chosen vector in $\mathbb{F}_{2^m}^{n}$ of weight $t$.
	
	\item $\dec_{\ME}(c,(S,G,P,\textbf{D}_{\mathcal{C}}))$: given the ciphertext $c$, compute $cP^{-1}$ and apply the decoding algorithm $\textbf{D}_{\mathcal{C}}$, obtaining the vector $x=mS$. The plaintext is given by $xS^{-1}=m$.
\end{itemize}

We recall that this scheme is not IND-CPA secure \cite{nojima2008semantic}, for stronger variants achieving this security level and others see \cite{dottling2012cca2,persichetti2018cca2}. Moreover, the state of the art uses matrices in systematic form as public key to reduce their size, an example can be Classic McEliece \cite{bernstein2017classic} (even if it is based on Niederreiter \cite{niederreiter1986knapsack}, a cryptosystem equivalent to McEliece).

We recall the modified version of the cryptosystem above from \cite{liu2019secure}. We refer to this scheme as Modified McEliece ($\MME$).
\begin{itemize}
	\item $\kgen_{\MME}(1^\lambda)$: generate a Goppa code $\mathcal{C}$ over $\mathbb{F}_{2^m}$ with parameters $[n=2^m,k,2t+1]$ according to $\lambda$. Let $G$ be a generator matrix of $\mathcal{C}$ and $\textbf{D}_{\mathcal{C}}$ an efficient decoding algorithm. Sample two random $k\times k$ invertible matrices $A$ and $B$, and a $n\times n$ permutation matrix $P$, all with coefficients in $\mathbb{F}_2$. Set $G'=AGP$ and $G''=BGP$. Moreover, set
	$$ \rho = (A+B)^{-1}B, \quad \gamma=\left[A+B\left(A+B\right)^{-1}B\right]^{-1}; $$
	if such matrices are non-invertible, pick different $A$ and $B$.
	Return $(G',G'',t)$ as public key and $(P,\rho,\gamma,\textbf{D}_{\mathcal{C}})$ as secret key.
	
	\item $\enc_{\MME}(m,(G',G'',t))$: given a message $m$ in $\mathbb{F}_2^{k}$, split it as $m=m_1+m_2$, where $m_1$ is random. The ciphertext of $m$ is given by $(c_1,c_2)$ where
	$$ c_1=m_1G'+m_2G''+e_1,\qquad c_2=mG'+m_1G''+e_2$$
	with $e_1,e_2$ random elements of $\mathbb{F}_2^n$ of weight $t$.
	
	\item $\dec_{\MME}(c,(P,\rho,\gamma,\textbf{D}_{\mathcal{C}}))$: given the ciphertext $c=(c_1,c_2)$, compute $c_1P^{-1}$ and $c_2P{-1}$. Decode these vectors with $\textbf{D}_{\mathcal{C}}$ and obtain $x_1,x_2$. The plaintext can be reconstructed by
	$$ m = (x_2+x_1\rho)\gamma.$$
\end{itemize}
	For an analysis of the correctness of this scheme, we remand to the original work \cite{liu2019secure}.
	
	\subsection{Security Analysis}
	
	The main point of the modification of $\MME$ is the use of two different public matrices $G'$ and $G''$, as well as the randomization in the encryption procedure, splitting the message into two random shares $m=m_1+m_2$.
	
	In the following result, we reduce the security of $\MME$ to the security of $\ME$.
	\begin{Prop}
		Let $\mathcal{F}_{\ME}$ be a key-recovery forger able to retrieve the secret key for the scheme $\ME$, then it is possible to design a key-recovery forger $\mathcal{F}_{\MME}$ for $\MME$ that uses $\mathcal{F}_{\ME}$ as a subroutine.
	\end{Prop}
	\begin{proof}
		Let $(P,\rho,\gamma,\textbf{D}_{\mathcal{C}})$ be the secret key of the public key $(G',G'',t)$ of the scheme $\MME$. By hypothesis let $\mathcal{F}_{\ME}$ be a forger for $\ME$: given a public key for $\ME$, it returns, with non-negligible probability, the corresponding secret key. We design the forger $\mathcal{F}_{\MME}$ as follows:
		\begin{enumerate}
			\item on input $(G',G'',t)$, observe that $G'=AGP$ and $G''=BGP$ are generator matrices of the same code, a permutation of the code generated by $G$. There exists a non-invertible matrix $\Sigma$ such that
			$$ G''=\Sigma G'.$$
			This implies $B=\Sigma A$.
			\item Use the forger $\mathcal{F}_{\ME}$ on input $(G',t)$, in fact this is a valid public key for $\ME$. The forger returns the secret key $(S,P,\textbf{D}_{\mathcal{C}})$ with $A=S$. Now we can compute $B$ and the rest of the $\MME$ secret key.
		\end{enumerate}
		Observe that the computation of $\Sigma$ involves only a Gauss elimination and therefore it can be carried in polynomial time. The computation of the $\MME$ secret key $(P,\rho,\gamma,\textbf{D}_{\mathcal{C}})$, given the secret key $(S,P,\textbf{D}_{\mathcal{C}})$ for $\ME$, involves only matrices operations. This implies that the key-recovery forger $\mathcal{F}_{\MME}$ is a polynomial-time algorithm with success probability equal to success probability of $\mathcal{F}_{\ME}$ and hence, non-negligible.
	\end{proof}
	
	The previous proposition shows how the security of the secret key of $\MME$ is at least as weak as the one of $\ME$. This is not a problem itself, but, if we take into account that the size of the public key of $\MME$ is the double of the one of $\ME$, and the fact that the major drawback of $\ME$ is the size of the public key, the adoption of $\MME$ does not solve this problem.
	
\section{The Signature Scheme}
\label{Sect:SS}
	In this section, we report and analyze the CFS-like signature scheme presented in \cite{liu2019secure}. The protocol, that we call $\LY$ is the following:
	\begin{itemize}
		\item $\kgen_{\LY}(1^\lambda)$: generate a Goppa code $\mathcal{C}$ over $\mathbb{F}_{2^m}$ with parameters $[n=2^m,k,2t+1]$ according to $\lambda$. Let $G$ be a generator matrix of $\mathcal{C}$ and $\textbf{D}_{\mathcal{C}}$ an efficient decoding algorithm. Sample two random $k\times k$ invertible matrices $A$ and $B$, and a $n\times n$ permutation matrix $P$, all with coefficients in $\mathbb{F}_2$. If the matrix $\left(A+B+AB^{-1}A\right)$ is singular, pick different $A$ and $B$. Set $G'=AGP$ and $G''=BGP$. Choose two hash functions $h_1,h_2:\{0,1\}^*\to \{0,1\}^n$.
		Return $(G',G'',t,h_1,h_2)$ as public key and $(A,B,P,\textbf{D}_{\mathcal{C}})$ as secret key.
		
		\item $\sig_{\LY}(M,(A,B,P,\textbf{D}_{\mathcal{C}}))$: given the message $M$, compute $d=h_1(M)$ and find two nonces $i_1,i_2$ in $\mathbb{N}$ such that $y_1=h_1(d||i_1)P^{-1}$ and $y_2=h_2(d||i_2)P^{-1}$ are two decodable vectors using $\textbf{D}_\mathcal{C}$. From $y_i$ we obtain messages $x_i=\textbf{D}_{\mathcal{C}}(y_i)$, for $i=1,2$. Set $e_1=y_1+x_1$ and $e_2=y_2+x_2$.
		Now set
		\begin{align*}
			m_1 &= \left(x_1+x_0 B^{-1}A\right) \left(A+B+AB^{-1}A\right)^{-1} \\
			m_2 &=x_0 B^{-1} + \left(x_1+x_0B^{-1}A\right) \left(A+B+AB^{-1}A\right)AB^{-1}.
		\end{align*}
		Return the tuple $(i_1,i_2,m_1,m_2,e_1,e_2)$ as a signature for $M$ .
		
		\item $\verifyx_{\LY}(M,\sigma,(G',G'',t,h_1,h_2))$: given the signature $\sigma$ for the message $M$, parse $\sigma$ as $(i_1,i_2,m_1,m_2,e_1,e_2)$ and check that
		\begin{enumerate}
			\item $m_1G' + m_2G'' +e_1 = h_1(h_1(M)||i_1)$ and
			\item $mG' + m_1G'' +e_2 = h_2(h_1(M)||i_2)$.
		\end{enumerate}
		The signature is valid if both checks pass, otherwise reject.
	\end{itemize}

	This scheme uses the CFS construction \cite{courtois2001achieve}, using brute-force hashing to find a decodable digest. While CFS does this process only once, $\LY$ searches two decodable digests. The probability of finding a single such digest is $\sim \frac{1}{t!}$ and then, the number of hashes and decodings in $\sig_{\LY}$ is roughly $2t!$.
	
	The authors of $\LY$, in their work \cite{liu2019secure}, propose some small parameters to improve the signing time of the scheme. Since they assume that the security of the signature scheme is based on the security of the $\MME$ encryption scheme, propose to set $t$ equal to 1 or 2. In this way the number of decoding attempts is very low, instead of $10!$ as proposed in the original CFS signature \cite{courtois2001achieve}. Unfortunately, this approach leads to some security issues, as shown in \cite{loidreau2001weak}. Indeed, for a small $t$, there exists a practical key-recovery attack enumerating all the irreducible Goppa codes and checking the permutation equivalence via the Support Splitting Algorithm \cite{sendrier2000finding}. This algorithms checks and returns, if any, the permutation between two codes. For a larger class of codes, and most notably on random codes, the Support Splitting Algorithm runs in polynomial time. This is the case of Goppa codes \cite{loidreau2001weak}. We reassume this reasoning in the following result.
	
	\begin{Prop}
		The signature scheme $\LY$, with parameters proposed in \cite{liu2019secure}, i.e. $m=16$ and $t=1,2$, admits a key-recovery forger.
	\end{Prop}
	\begin{proof}
		A key-recovery forger can be defined as follows. Given the public key $(G',G'',t)$, using the Support Splitting Algorithm and enumerating all $2^{mt}/t$ irreducible Goppa codes we can find a permutation equivalent code to the one used in the secret key. The number of codes to check can be decreased up to $2^{m(t-3)}/mt$ using techniques from \cite{loidreau2001weak}. The Support Splitting Algorithm has average complexity $O(n^3)$ and, for small $t$, this leads to a practical attack to the secret key of the scheme: finding a permutation equivalent Goppa code allows to recover the secret key.
	\end{proof}

	We want to highlight that using a bigger $t$ could ensure the security of the scheme but does not solve the inefficiency problem of CFS. Moreover, in the case of $\LY$, the number of decoding attempts is twice of what is done in CFS, resulting in an even slower signing algorithm. This does not justify the adoption of $\LY$ over CFS.
	
	\section{Conclusions}
	\label{Sect:concl}
	The public key cryptosystem $\MME$ is not weaker than the standard scheme $\ME$, but the adoption of a longer public key does not bring any additional security. The digital signature $\LY$ has the same issue as CFS: for a reasonable level of security, it is unpractical. The field of code-based signature schemes is very prosperous, and the research is still looking for other solutions which maintain the original security of CFS but also decrease the computational complexity of the signing algorithm.

	\section*{Acknowledgments}
	The author is a member of the INdAM Research group GNSAGA and acknowledges support from TIM S.p.A. through the PhD scholarship.

\bibliographystyle{spmpsci}
\bibliography{local}

\begin{thebibliography}{10}
\providecommand{\url}[1]{{#1}}
\providecommand{\urlprefix}{URL }
\expandafter\ifx\csname urlstyle\endcsname\relax
  \providecommand{\doi}[1]{DOI~\discretionary{}{}{}#1}\else
  \providecommand{\doi}{DOI~\discretionary{}{}{}\begingroup
  \urlstyle{rm}\Url}\fi

\bibitem{barenghi2021less}
Barenghi, A., Biasse, J.F., Persichetti, E., Santini, P.: Less-fm: fine-tuning
  signatures from the code equivalence problem.
\newblock In: International Conference on Post-Quantum Cryptography, pp.
  23--43. Springer (2021)

\bibitem{berlekamp1973goppa}
Berlekamp, E.: Goppa codes.
\newblock IEEE Transactions on Information Theory \textbf{19}(5), 590--592
  (1973)

\bibitem{berlekamp1978inherent}
Berlekamp, E., McEliece, R., Van~Tilborg, H.: On the inherent intractability of
  certain coding problems (corresp.).
\newblock IEEE Transactions on Information Theory \textbf{24}(3), 384--386
  (1978)

\bibitem{bernstein2017classic}
Bernstein, D.J., Chou, T., Lange, T., von Maurich, I., Misoczki, R.,
  Niederhagen, R., Persichetti, E., Peters, C., Schwabe, P., Sendrier, N.,
  et~al.: Classic mceliece: conservative code-based cryptography.
\newblock NIST submissions  (2017)

\bibitem{courtois2001achieve}
Courtois, N.T., Finiasz, M., Sendrier, N.: How to achieve a mceliece-based
  digital signature scheme.
\newblock In: International Conference on the Theory and Application of
  Cryptology and Information Security, pp. 157--174. Springer (2001)

\bibitem{d2021security}
D'Alconzo, G., Meneghetti, A., Piasenti, P.: Security issues of cfs-like
  digital signature algorithms.
\newblock arXiv preprint arXiv:2112.00429  (2021)

\bibitem{debris2019wave}
Debris-Alazard, T., Sendrier, N., Tillich, J.P.: Wave: A new family of trapdoor
  one-way preimage sampleable functions based on codes.
\newblock In: International Conference on the Theory and Application of
  Cryptology and Information Security, pp. 21--51. Springer (2019)

\bibitem{dottling2012cca2}
Dottling, N., Dowsley, R., Muller-Quade, J., Nascimento, A.C.: A cca2 secure
  variant of the mceliece cryptosystem.
\newblock IEEE Transactions on Information Theory \textbf{58}(10), 6672--6680
  (2012)

\bibitem{feneuil2021shared}
Feneuil, T., Joux, A., Rivain, M.: Shared permutation for syndrome decoding:
  New zero-knowledge protocol and code-based signature.
\newblock Cryptology ePrint Archive  (2021)

\bibitem{gueron2022designing}
Gueron, S., Persichetti, E., Santini, P.: Designing a practical code-based
  signature scheme from zero-knowledge proofs with trusted setup.
\newblock Cryptography \textbf{6}(1), 5 (2022)

\bibitem{kabatianskii1997digital}
Kabatianskii, G., Krouk, E., Smeets, B.: A digital signature scheme based on
  random error-correcting codes.
\newblock In: IMA International Conference on Cryptography and Coding, pp.
  161--167. Springer (1997)

\bibitem{liu2019secure}
Liu, X., Yang, X., Han, Y., Wang, X.A.: A secure and efficient code-based
  signature scheme.
\newblock International Journal of Foundations of Computer Science
  \textbf{30}(04), 635--645 (2019)

\bibitem{loidreau2001weak}
Loidreau, P., Sendrier, N.: Weak keys in the mceliece public-key cryptosystem.
\newblock IEEE Transactions on Information Theory \textbf{47}(3), 1207--1211
  (2001)

\bibitem{macwilliams1977theory}
MacWilliams, F.J., Sloane, N.J.A.: The theory of error correcting codes,
  vol.~16.
\newblock Elsevier (1977)

\bibitem{mceliece1978public}
McEliece, R.J.: A public-key cryptosystem based on algebraic.
\newblock Coding Thv \textbf{4244}, 114--116 (1978)

\bibitem{niederreiter1986knapsack}
Niederreiter, H.: Knapsack-type cryptosystems and algebraic coding theory.
\newblock Prob. Contr. Inform. Theory \textbf{15}(2), 157--166 (1986)

\bibitem{nojima2008semantic}
Nojima, R., Imai, H., Kobara, K., Morozov, K.: Semantic security for the
  mceliece cryptosystem without random oracles.
\newblock Designs, Codes and Cryptography \textbf{49}(1), 289--305 (2008)

\bibitem{patterson1975algebraic}
Patterson, N.: The algebraic decoding of goppa codes.
\newblock IEEE Transactions on Information Theory \textbf{21}(2), 203--207
  (1975)

\bibitem{persichetti2018cca2}
Persichetti, E.: On the cca2 security of mceliece in the standard model.
\newblock In: International Conference on Provable Security, pp. 165--181.
  Springer (2018)

\bibitem{sendrier2000finding}
Sendrier, N.: Finding the permutation between equivalent linear codes: The
  support splitting algorithm.
\newblock IEEE Transactions on Information Theory \textbf{46}(4), 1193--1203
  (2000)

\end{thebibliography}
\end{document}